\documentclass{amsart}
\usepackage{enumerate}
\usepackage{amssymb}
\usepackage{latexsym,amscd,amsfonts}
\def \Z {\mathbb Z}
\def \R {\mathbb R}
\def \C {\mathbb C}

\def \N {\mathbb N}

\def\dd{\mathrm d}

\newcommand{\SpT}{\mathrm{Sp}_{3}(\R)\,}
\newcommand{\SpN}{\mathrm{Sp}_{\mathrm{N}}(\R)\,}
\newcommand{\supp}{\mathrm{supp}\,}
\newcommand{\clz}{\mathrm{Cl}_{\mathrm{Z}}(G(E))}
\newcommand{\gt}{\mathfrak{g}_{3}(E)}
\newcommand{\spt}{\mathfrak{sp}_{3}(\R)}
\newcommand{\gtt}{\tilde{\mathfrak{g}}_{3}(E)}
\newtheorem{thm}{Theorem}
\newtheorem{prop}{Proposition}

\newtheorem{Def}{Definition}
\newtheorem{lem}{Lemma}
\begin{document}
%-------------------------------------------------------------%
       
%-------------------------------------------------------------%
\title{A matrix-valued point interactions model}
\author{Hakim Boumaza} 
\email{boumaza@math.keio.ac.jp}
%\runningtitle{}
%\runningauthor{}
\address{Keio University, Department of Mathematics\\
Hiyoshi 3-14-1, Kohoku-ku 223-8522\\
Yokohama, Japan\\}
\thanks{The author is supported by JSPS Grant P07728}
%\date{\today}
%-------------------------------------------------------------%
\begin{abstract}
We study a matrix-valued Schr\"odinger operator with random point interactions. We prove the absence of absolutely continuous spectrum for this operator by proving that away from a discrete set its Lyapunov exponents do not vanish. For this we use a criterion by Gol'dsheid and Margulis and we prove the Zariski denseness, in the symplectic group, of the group generated by the transfer matrices. Then we prove estimates on the transfer matrices which lead to the H\"older continuity of the Lyapunov exponents. After proving the existence of the integrated density of states of the operator, we also prove its H\"older continuity by proving a Thouless formula which links the integrated density of states to the sum of the positive Lyapunov exponents.
\end{abstract}


%\classification{Mathematics Subject Classification 2000}{47B80, 37H15}
%-------------------------------------------------------------%
\maketitle
%-------------------------------------------------------------%
%:1.Introduction
%-------------------------------------------------------------%
\section{Introduction : model and results}

\noindent This paper discusses properties of the Lyapunov exponents and the of integrated density of states of the following formal Schr\"odinger operator with random point interactions,
\begin{equation}\label{formalop}
H_P(\omega)=-\frac{\mathrm{d}^{2}}{\mathrm{d} x^{2}}I_{\mathrm{N}}+ V_0 +\sum_{n\in \Z} \left(
\begin{array}{ccc}
c_1 \omega_{1}^{(n)} \delta_{0}(x-n) & & 0 \\
& \ddots & \\
0 & &c_N \omega_{N}^{(n)} \delta_{0}(x-n)
\end{array}\right)
\end{equation}
acting on $L^2(\R)\otimes \C^N$. Here, $V_0$ is the maximal multiplication operator by the constant coefficient matrix with $0$ on its diagonal, $1$ on the upper and lower diagonals and $0$ everywhere else. Moreover, $c_1,\ldots,c_N$ are real numbers, $\delta_0$ is the Dirac distribution at the point $0$ and $I_{\mathrm{N}}$ is the identity matrix of order $N$, $N\geq 1$. Let $(\Omega, \mathcal{A}, \mathsf{P})$ be a complete probability space on which we define the sequence $(\omega^{(n)})_{n\in \Z}=((\omega_1^{(n)},\ldots,\omega_N^{(n)}))_{n\in \Z}$ of independent and identically distributed (\emph{i.i.d.}) $\R^N$-valued random variables of common distribution $\nu$. We assume that the support of $\nu$, $\supp \nu \subset \R^N$, is bounded and that $\{x-y\ |\ x,y\in \supp \nu \}$ spans $\R^N$. An example of such a distribution is obtained when the components $\omega_1^{(n)},\ldots,\omega_N^{(n)}$ are \emph{i.i.d.} real-valued random variables such that the support of their common distribution contains at least two different points (like in the case of Bernoulli variables). We also set $(\tilde{\Omega}, \tilde{\mathcal{A}}, \tilde{\mathsf{P}})=\bigotimes_{n\in \Z} (\Omega, \mathcal{A}, \mathsf{P})$.

\noindent Following \cite{albeverio}, we define the self-adjoint realization of the formal operator $H_P(\omega)$, for every $\omega \in \tilde{\Omega}$,  by
\begin{equation}\label{opH}
H(\omega)=\bigoplus_{i=1}^{N} H_{\omega_i} + V_0
\end{equation}
acting on $L^2(\R)\otimes \C^N$. Each $H_{\omega_i}$ for $i\in \{1,\ldots ,N\}$ is the operator acting on $L^2(\R)\otimes \C$ by $H_{\omega_i}f=-f''$ whose domain is 
\begin{eqnarray}\label{Hoidomain}
D(H_{\omega_i}) & = & \{ f\in L^2(\R)\otimes \C\ |\  f,f'\ \mathrm{are}\ \mathrm{absolutely}\ \mathrm{continuous}\ \mathrm{on}\ \R\setminus \Z,  \nonumber \\
& & f'' \in L^2(\R)\otimes \C, f \ \mathrm{is}\ \mathrm{continuous} \ \mathrm{on}\ \R, \nonumber \\
& & f'(n^+)=f'(n^-) + c_i \omega_i^{(n)} f(n)\ \mathrm{for}\ \mathrm{every}\ n\in\Z\}, \nonumber
\end{eqnarray}
where existence of the left and right derivatives $f'(n^-)$ and $f'(n^+)$ at all integers is assumed. The $\omega_i$'s, $V_0$ and the $c_i$'s are as above. These operators are self-adjoint and bounded from below (see \cite{albeverio}) as the support of the distribution $\nu$ is bounded. Thus, $V_0$ being bounded and self-adjoint, it implies that $H(\omega)$ is self-adjoint for every $\omega \in \tilde{\Omega}$. Moreover, $H(\omega)$ is a $\Z$-ergodic operator and therefore its almost-sure spectrum is included in $\R$.

\noindent The operator $H(\omega)$ is the Hamiltonian which describe the propagation of an electron in $N$ coupled one-dimensional layers. The random variables are describing point sources that interact with the electron. These point sources are situated at each integer point of the system.

\noindent In the case $N=2$, a previous work from G\"unter Stolz and the author (see \cite{stolzboumaza}) stated that $H(\omega)$ has almost-surely no absolutely continuous spectrum. We now look at the case $N=3$ which is our first result on model $(\ref{opH})$.

\begin{thm}\label{th1}
We assume that $N=3$ and that $c_2\in \R$, $c_1$ and $c_2$ are non-zero. Then there exists a discrete set $\mathcal{S} \subset \R$ such that for every $E\in \R \setminus \mathcal{S}$, the Lyapunov exponents associated to $H(\omega)$ at $E$ are all distinct and positive,
$$\gamma_1(E) > \gamma_2(E) > \gamma_3(E) >0.$$
Thus, $H(\omega)$ has almost-surely no absolutely continuous spectrum.
\end{thm}

\noindent The approach to prove this theorem will be to study the F\"urstenberg group $G(E)$ associated to every real number $E$ to $H(\omega)$. That is, the group generated by the transfer matrices associated to $H(\omega)$. We will recall precise definitions and properties of the transfer matrices, the F\"urstenberg group and the Lyapunov exponents in Section \ref{sec:lyapprop}. Let $\SpN$ denote the group of $2N\times 2N$ symplectic matrices. The F\"urstenberg group $G(E)$ is a subgroup of $\SpN$. For the definitions of $L_p$-strong irreducibility and $p$-contractivity we refer to \cite{bougerol}, definitions $A.IV.3.3$ and $A.IV.1.1$, respectively. Then the proof of theorem \ref{th1} is based upon the following criterion due to Gol'dsheid and Margulis.

\begin{thm}[\cite{GM89}]\label{thmGM}
If the group $G(E)$ is Zariski dense in $\SpN$ then it is $p$-contracting and $L_p$-strongly irreducible for every $p\in \{1,\ldots,N \}$. Thus the Lyapunov exponents associated to $H(\omega)$ are distinct and positive.
\end{thm}

\noindent Once we have proved the Zariski denseness of the F\"urstenberg group, by Kotani's theory (see \cite{KS88}) we will obtain the almost-surely absence of absolutely continuous spectrum of $H(\omega)$. Actually \cite{KS88} does not cover directly the case of point interactions models as (\ref{opH}) but we will see in Section \ref{sec:proofth4} how to adapt proofs of \cite{KS88} to our setting.
\vskip 2mm

\noindent Our second result will be on the regularity of the Lyapunov exponents viewed as functions of the real parameter $E$. We first prove a general result of H\"older continuity of the Lyapunov exponents associated to $H(\omega)$ under suitable assumptions on $G(E)$.

\begin{thm}\label{th2}
Let $I$ be a compact interval in $\R$. We assume that the group $G(E)$ is $p$-contracting and $L_{p}$-strongly irreducible for every $p\in \{1,\ldots ,N\}$ and all $E\in I$. Then the Lyapunov exponents associated to $H(\omega)$ are H\"older continuous on $I$, \emph{i.e}, there exist two real numbers $\alpha >0$ and $0<C<+\infty$ such that 
$$\forall p\in \{1,\ldots N\},\ \forall E,E' \in I,\ |\gamma_{p}(E)-\gamma_{p}(E')| \leq C |E-E'|^{\alpha}.$$
\end{thm}

The proof of theorem \ref{th2} is mostly based upon the existence of an integral representation of the Lyapunov exponents under the assumptions of $p$-contractivity and $L_p$-strong irreducibility of $G(E)$. It will also require estimates on the transfer matrices that will be proved in Section \ref{sec:regularlyap}. Combining theorem \ref{th1} and theorem \ref{th2} we also obtain,

\begin{thm}\label{th3}
We assume that $N=3$ and that $c_2\in \R$, $c_1$ and $c_2$ are non-zero. Let $\mathcal{S}$ be the discrete subset of $\R$ obtained in theorem \ref{th1}. Let $I\subset \R \setminus \mathcal{S}$ be a compact interval. Then the Lyapunov exponents associated to $H(\omega)$ are H\"older continuous on $I$.
\end{thm}
\vskip 2mm

\noindent Finally, in Section \ref{sec:regularids}, we turn to the study of the integrated density of states. It is a function of one real variable that counts the mean number per unit volume of spectral values of $H(\omega)$ below a fixed real number $E$. For operators like $H(\omega)$, acting on an infinite dimensional Hilbert space, this function could be equal to infinity for every real number $E$. To avoid this problem, we will first restrict the operator to intervals of $\R$ of finite length. Let $L$ be a strictly positive integer and $D=[-L,L]\subset \R$. We set $H^{(D)}(\omega)$ the restriction of $H(\omega)$ acting on $L^{2}(D)\otimes \C^N$ with Dirichlet boundary conditions on $D$. 

\begin{Def}
The integrated density of states, or IDS, associated to $H(\omega)$ is the function $N:\R \to \R^+$ defined for each $E\in \R$ as the following thermodynamical limit,
\begin{equation}\label{IDSdefi}
N(E)=\lim_{L\to +\infty} \frac{1}{2L} \# \{\lambda \leq E\ |\ \lambda \in \sigma(H^{(D)}(\omega)) \}.
\end{equation}
\end{Def}

\noindent In Section \ref{sec:existids} we will start by proving the existence of this limit. For this, we will have to prove a matrix-valued Feynman-Kac formula for matrix-valued point interactions models. To prove the existence of the IDS for every $E \in \R$, we will adapt ideas of \cite{FM2003} for scalar-valued point interactions models combined with what was done in \cite{boumazarmp} for matrix-valued Anderson models. Then in Section \ref{sec:proofth4} we deduce a Thouless formula by adapting Kotani's theory to point interactions models. This formula states that the sum of positive Lyapunov exponents and the integrated density of states are harmonically conjugated. Combining this Thouless formula, properties of the Hilbert transform and theorem \ref{th2}, we prove the following result.

\begin{thm}\label{th4}
Let $I$ be a compact interval in $\R$. We assume that the group $G(E)$ is $p$-contracting and $L_{p}$-strongly irreducible for every $p\in \{1,\ldots ,N\}$ and all $E\in I$. Then the integrated density of states associated to $H(\omega)$ is H\"older continuous on $I$.
\end{thm}

\noindent Combining theorem \ref{th4} with theorem \ref{th1} we get this last theorem.

\begin{thm}\label{th5}
We assume that $N=3$ and that $c_2\in \R$, $c_1$ and $c_3$ are non-zero. Let $\mathcal{S}$ be the discrete subset of $\R$ obtained in theorem \ref{th1}. Let $I\subset \R \setminus \mathcal{S}$ be a compact interval. Then the integrated density of states associated to $H(\omega)$ is H\"older continuous on $I$.
\end{thm}

\noindent The fact that we can choose $c_2=0$ in theorem \ref{th1} and theorem \ref{th5} means that we are here in the presence of a phenomenon of propagation of randomness similar to the one found in the work of Glaffig (see \cite{glaffig}). Let us explain what we observe here. If $c_2=0$ or, as we will see in Section \ref{sec:proofth1}, if $H_{\omega_2}$ is deterministic, we still have the positivity of all the Lyapunov exponents and the regularity of the IDS of $H(\omega)$. But, due to Kotani's theory, the positivity of the Lyapunov exponents is directly related to the randomness in the model. Heuristically, if the second layer in our model (corresponding to $H_{\omega_2}$) is deterministic, we should not have that all Lyapunov exponents are positive. But the fact that we have randomness on the first and the third layer and that the three layers are coupled by the action of $V_0$ seems to mean that the randomness is somehow transported to the second layer. To the best of our knowledge the only other example of such a phenomenon of propagation of randomness  can be found in \cite{glaffig} where Glaffig proves a strong regularity result for the IDS of a discrete matrix-valued Schr\"odinger operator. In Glaffig's model, the randomness on the first layer is assumed to be very strong as it is given by \emph{i.i.d.} random variables whose common law is very regular (absolutely continuous with respect to the Lebesgue measure and with a Radon-Nikodym derivative being in a Sobolev space). On the other layers he only assumed that the random variables follow a Bernoulli law. Despite that lack of regularity on the other layers, he still proves that the IDS is $C^{\infty}$ which should have occured intuitively only if all the random variables were as regular as those on the first layer. But, as all the layers are coupled in his model, the strong regularity on the first layer is somehow propagating to the other layers.

\noindent Even if we do not prove that the IDS in our model is $C^{\infty}$ like in Glaffig's model, we prove the H\"older continuity with even less randomness than in Glaffig's model. We only need Bernoulli variables on the first and the third layers and we can even assume that the second layer is deterministic. Still, the H\"older continuity of the IDS is sufficient to hope to be able to prove in the future the Anderson localization for $H(\omega)$, that is that the spectrum of $H(\omega)$ is pure point and the associated eigenfunctions are exponentially decaying to $0$ at infinity. Indeed, with the H\"older continuity of the IDS we should be able to prove a Wegner estimate which is the key ingredient in order to apply a multi-scale analysis scheme to prove Anderson localization (see \cite{stollmann, kleingerminet}).
\vskip 2mm

\noindent We finish this introduction by giving the outline of the rest of the paper. Section \ref{sec:lyap} will be devoted to the proof of theorem \ref{th1}. In Section \ref{sec:lyapprop} we will recall the definitions of the Lyapunov exponents, the transfer matrices and the F\"urstenberg group. Then in Section \ref{sec:proofth1} we will present the computations which leads to the Zariski denseness of the F\"urstenberg group in the symplectic group. These computations will be very similar to those done in \cite{GM89} for the one-dimensional discrete matrix-valued Schr\"odinger operator. The reason is that in the point interactions model that we study here, the random parameters only appear at integer points. Thus, in the transfer matrices, the random parameters will appear in the same way as in the discrete model. The difference will be that we have here a continuous model and thus the energy parameter $E$ will appear in a way different than in the discrete case. This is the reason of the existence of the set $\mathcal{S}$ of critical energies in theorem \ref{th1}. These critical energies did not appear in the discrete model of \cite{GM89}. 

\noindent In Section \ref{sec:regular} we will focus on the regularity results for the Lyapunov exponents and the IDS. Section \ref{sec:regularlyap} deals with the regularity of the Lyapunov exponents while in Section \ref{sec:regularids} we will prove the existence of the IDS and its H\"older continuity.

%-------------------------------------------------------------%
%:2.Proof of theorem 1
%-------------------------------------------------------------%
\section{Positivity of the Lyapunov exponents}\label{sec:lyap}

\subsection{Lyapunov exponents and transfer matrices}\label{sec:lyapprop}

We start this section with a review of some results about Lyapunov exponents. These results holds for general sequences of \emph{i.i.d.} random symplectic matrices. Let $N$ be a positive integer and $\SpN$ denote the group of $2N\times 2N$ real symplectic matrices. It is the subgroup of $\mathrm{GL}_{2\mathrm{N}}(\R)$ of matrices $M$ satisfying $ ^tMJM=J$, where $J$ is the matrix of order $2N$ defined by $J=\left(\begin{array}{cc}
0 & -I_{\mathrm{N}} \\
I_{\mathrm{N}} & 0
\end{array}\right)$.

\begin{Def}
Let $(A_{n}^{\omega})_{n\in \N}$ be a sequence of \emph{i.i.d.}\ random matrices in $\SpN$ with $\mathbb{E}(\log^{+}||A_{0}^{\omega}||) <+\infty$. The Lyapunov exponents $\gamma_{1},\ldots,\gamma_{2N}$ associated with $(A_{n}^{\omega})_{n\in \N}$ are defined inductively, for every $p\in \{1,\ldots,2N\}$, by
\begin{equation}\label{lyapdef}
\sum_{i=1}^{p} \gamma_{i} = \lim_{n \to \infty} \frac{1}{n}
\mathbb{E}(\log ||\wedge^{p} (A_{n-1}^{\omega}\ldots A_{0}^{\omega})
||).
\end{equation}
\end{Def}

\noindent Here, $\wedge^{p} M$ denotes the $p$th exterior power of the matrix $M$, acting on the $p$th exterior power of $\R^{2N}$. One has $\gamma_{1}\geq \ldots \geq \gamma_{2N}$. Moreover, the random matrices $(A_{n}^{\omega})_{n\in \N}$ being symplectic, we have the symmetry property $\gamma_{2N-i+1}= -\gamma_{i}$, for every $i \in \{1,\ldots,N\}$ (see \cite{bougerol} p.$89$, Prop $3.2$).

\noindent To define Lyapunov exponents associated to the operator $H(\omega)$, we first introduce the sequence of transfer matrices associated to $H(\omega)$. Let $E\in \R$ and consider the second order differential system
\begin{equation}\label{system2}
H(\omega)u=Eu.
\end{equation}
A function $u=(u_1,\ldots,u_N):\R \to \C^N$ is called a solution of (\ref{system2}) whenever $-u''+V_0 u=Eu$ on $\R \setminus \Z$ and each $u_i$ satisfies the same boundary conditions as elements in $D(H_{\omega_i})$, that is 
\begin{equation}\label{solsystem2}
\forall i\in \{1,\ldots, N\},\ \forall n\in \Z,\ u_{i}'(n^+)=u_{i}'(n^-)+c_i \omega_{i}^{(n)} u_i (n).
\end{equation}

\begin{Def}
If $u$ is a solution of (\ref{system2}) then the transfer matrix $A_{(n,n+1]}^{\omega^{(n)}}(E)$ from $n^+$ to $(n+1)^+$ is defined by the relation
\begin{equation}\label{reltransfermat}
\left( \begin{array}{c}
u((n+1)^+) \\
u'((n+1)^+) 
\end{array} \right) = A_{(n,n+1]}^{\omega^{(n)}}(E) \left( \begin{array}{c}
u(n^+) \\
u'(n^+) 
\end{array} \right).
\end{equation}
\end{Def}

\noindent Then for every $n\in \Z$, $A_{(n,n+1]}^{\omega^{(n)}}(E) \in \SpN$ as $ ^t(u,u')$ is a solution of the first order Hamiltonian system associated to system (\ref{system2}). Thus the sequence $\left(A_{(n,n+1]}^{\omega^{(n)}}(E)\right)_{n\in \Z}$ is a sequence of \emph{i.i.d.} symplectic matrices and the Lyapunov exponents associated to it are by definition the Lyapunov exponents of $H(\omega)$.
\vskip 2mm

\noindent We have an explicit form for $A_{(n,n+1]}^{\omega^{(n)}}(E)$. To compute it we start by solving the free system (\ref{system2}) on $(n,n+1)$. Actually we only have to do it on $(0,1)$ due to the $1$-periodicity of $V_0$. Then the transfer matrix from $n^+$ to $(n+1)^-$ is given by
\begin{equation}\label{transfer_exp} 
A_{(0,1)}(E)=\exp \left( \begin{array}{cc}
0 & I_{\mathrm{N}} \\
V_0-EI_{\mathrm{N}} & 0
\end{array}\right).
\end{equation}
We also set, for every $N\times N$ matrix $Q$ the $2N\times 2N$ matrix $M(Q)$ given by $M(Q)=\left( \begin{array}{cc}
I_{\mathrm{N}} & 0 \\
Q & I_{\mathrm{N}}
\end{array}\right)$. Then, using the interface relation (\ref{solsystem2}), the transfer matrix from $(n+1)^-$ to $(n+1)^+$ is $M(\mathrm{diag}(c_1\omega_1^{(n)},\ldots,c_N\omega_N^{(n)}))$. Thus we have
\begin{equation}\label{transfer_mat_split}
A_{(n,n+1]}^{\omega^{(n)}}(E)=M(\mathrm{diag}(c_1\omega_1^{(n)},\ldots,c_N\omega_N^{(n)})) \; A_{(0,1)}(E).
\end{equation}
The transfer matrix splits in a product of two factors. The first factor contains the random part of the transfer matrix and is independent of $E$. The second factor is deterministic and depends only on $E$.

\noindent As the matrices $A_{(n,n+1]}^{\omega^{(n)}}(E)$ are \emph{i.i.d.}, we denote by $\mu_{E}$ their common law.

\begin{Def}\label{f_group_def}
The F\"urstenberg group associated to $H(\omega)$ is the closed group generated by the support of $\mu_{E}$ : $G(E)=\overline{<\supp \mu_{E}>}$.
\end{Def}
Because the transfer matrices are \emph{i.i.d.} we have an internal description of $G(E)$,
\begin{equation}\label{f_group_interne}
G(E)=\overline{<A_{(0,1]}^{\omega^{(0)}}(E) \ |\ \omega^{(0)} \in \supp \nu>}.
\end{equation}

\subsection{Proof of theorem \ref{th1}}\label{sec:proofth1}

\noindent Using the criterion of Gol'dsheid and Margulis and Kotani's theory, theorem \ref{th1} reduces to the following proposition.

\begin{prop}\label{zariski_f_group}
Let $N=3$ and assume $c_2\in \R$ and $c_1$ and $c_3$ being non-zero. There exists a discrete subset $\mathcal{S}$ of $\R$ such that for every $E\in \R \setminus \mathcal{S}$, $G(E)$ is Zariski dense in $\SpT$.
\end{prop}

\noindent The rest of this section is devoted to the proof of this proposition. We begin by writing the explicit form of $A_{(0,1]}^{\omega^{(0)}}(E)$. For this we have to compute the exponential in $A_{(0,1)}(E)$. We start by diagonalizing in orthonormal basis the symmetric matrix defining $V_{0}$ :
\begin{equation}\label{V0diag}
V_0=\left( \begin{array}{ccc}
0 & 1 & 0 \\
1 & 0 & 1 \\
0 & 1 & 0
\end{array} \right)= U \left( \begin{array}{ccc}
1 & 0 & 0\\
0 & \sqrt{2} & 0\\
0 & 0 & -\sqrt{2}
\end{array} \right) U^{-1}
\end{equation}
with 
\begin{equation}\label{V0diagU}
U=\frac{1}{2} \left( \begin{array}{ccc}
-\sqrt{2} & 1 & 1 \\
0 & \sqrt{2} & -\sqrt{2}\\
\sqrt{2} & 1 & 1 
\end{array} \right)\ \mathrm{and}\ U^{-1}= ^tU.
\end{equation}
\textbf{We now assume that $E>\sqrt{2}$} and we will deal with the other cases later. By computing the successive powers of $\left( \begin{array}{cc}
0 & I_{\mathrm{N}} \\
V_0-EI_{\mathrm{N}} & 0
\end{array}\right)$ one gets
\begin{equation}\label{diagA01}
A_{(0,1)}(E)=\left( \begin{array}{cc}
U & 0 \\
0 & U
\end{array} \right) \; R_{\alpha,\beta,\gamma} \; \left( \begin{array}{cc}
U^{-1} & 0 \\
0 & U^{-1}
\end{array} \right)\;,
\end{equation}
where $\alpha=\sqrt{E-1}$, $\beta=\sqrt{E-\sqrt{2}}$, $\gamma=\sqrt{E+\sqrt{2}}$ and
\begin{equation}\label{ralphabetagamma}
R_{\alpha,\beta,\gamma}=\left( \begin{array}{cccccc}
\cos\alpha & 0 & 0 & \frac{1}{\alpha}\sin\alpha & 0 & 0 \\
0 & \cos\beta & 0 & 0 & \frac{1}{\beta}\sin\beta & 0 \\
0 & 0 & \cos\gamma & 0 & 0 & \frac{1}{\gamma}\sin\gamma \\
-\alpha \sin\alpha & 0 & 0 & \cos\alpha & 0 & 0 \\
0 & -\beta \sin\beta & 0 & 0 & \cos\beta & 0 \\
0 & 0 & -\gamma \sin\gamma & 0 & 0 & \cos\gamma
\end{array}\right)\;.
\end{equation}

\noindent To prove the Zariski denseness of $G(E)$ in $\SpT$ we use the fact that $\SpT$ is a connected Lie group. If $\clz$ denotes the Zariski closure of $G(E)$ in $\SpT$, we only have to show that $\gt$, the Lie algebra of $\clz$, is equal to $\spt$, the Lie algebra of $\SpT$. We recall that
\begin{equation}
\spt = \left\lbrace \left(\begin{array}{cc}
a & b_{1} \\
b_{2} & -{}^ta
\end{array} \right),\ a\in \mathcal{M}_{3}(\R),\ b_{1}\ \mathrm{and}\ b_{2}\ \mathrm{symmetric}\ \right\rbrace
\end{equation}
and $\spt$ is of dimension $21$. Our strategy will be to exhibit a family of $21$ linearly independent elements in $\gt$. Before starting this construction we prove the following lemma.

\begin{lem}\label{MQlemma}
For a three-by-three matrix $Q$ one has $M(Q) \in \clz$ if and only if $\left( \begin{array}{cc} 0 & 0 \\ Q & 0 \end{array} \right) \in \gt$.
\end{lem}

\begin{proof}
Assume that $\left( \begin{array}{cc} 0 & 0 \\ Q & 0 \end{array} \right) \in \gt$. Then $M(Q)=\exp(M(Q)-I)\in \clz$ because $\clz$ is a Lie group. Conversely, if $M(Q)\in \clz$ we consider the subgroup $G_Q := \{M(nQ) = M(Q)^n:n\in\Z\}$ of $\clz$. It follows that $M(xQ) \in \mathrm{Cl}_{\mathrm{Z}}(G_Q)$ for all $x\in\R$. To see this, let $p$ be a polynomial in $6\times 6$ variables such that $p(A)=0$ for all $A\in G_Q$. Then the polynomial in one variable $\tilde{p}(x) := p(M(xQ))$ vanishes at every integer point, thus it vanishes identically on $\R$. So, for every $x\in \R$, $p(M(xQ))=0$. Then by definition of the Zariski closure, for every $x\in \R$, $M(xQ)\in \mathrm{Cl}_{\mathrm{Z}}(G_Q)\subset \clz$. By writing 
$$M(xQ)=I_{\mathrm{2N}}+x\left( \begin{array}{cc} 0 & 0 \\ Q & 0 \end{array} \right)$$
and by differentiating at the identity (at $x=0$), by definition of a Lie algebra, $\left( \begin{array}{cc} 0 & 0 \\ Q & 0 \end{array} \right) \in \gt$.
\end{proof}

\begin{proof}[\textbf{Proof of proposition \ref{zariski_f_group} for $E>\sqrt{2}$}]
\noindent \textbf{Step 1 :} From (\ref{transfer_mat_split}), for $\omega^{(0)}$ and $\tilde{\omega}^{(0)}$ in $\supp \nu$ we have
\begin{equation}\label{aainv}
A_{(0,1]}^{\tilde{\omega}^{(0)}}(E) A_{(0,1]}^{\omega^{(0)}}(E)^{-1}= M(\mathrm{diag}(c_1(\tilde{\omega}_1^{(0)}-\omega_1^{(0)}),\ldots,c_3(\tilde{\omega}_3^{(0)}-\omega_3^{(0)}))) 
\end{equation}
is in $G(E)$. Thus, $\gt$ being an algebra and therefore closed under linear combinations, using lemma \ref{MQlemma} and the hypothesis that $\{x-y\ |\ x,y\in \supp \nu \}$ is spanning $\R^3$, we get that  $\left( \begin{array}{cc} 0 & 0 \\ Q & 0 \end{array} \right) \in \gt$ for every diagonal matrix $Q$. If we assume that $c_2=0$ then it is true only for $Q$ of the form $\mathrm{diag}(a_1,0,a_2)$, $a_1,a_2\in \R$.

\noindent \textbf{Step 2 :} By step $1$ and lemma \ref{MQlemma}, $M(Q)\in \clz$ for every diagonal matrix $Q$ (or at least with a zero on the second diagonal term if $c_2=0$). In particular,
\begin{equation}\label{A01_in_clz}
A_{(0,1)}(E) = M(\mathrm{diag}(c_1\omega_1^{(0)},\ldots, c_3\omega_3^{(0)}))^{-1} A_{(0,1]}(E) \in \clz.
\end{equation}

\noindent \textbf{Step 3 :} By a general property in Lie groups, $XMX^{-1} \in \gt$ for every $M\in \gt$ and $X\in \clz$. Thus, for every $l\in \Z$,
\begin{eqnarray}
\left( \begin{array}{cc}
U & 0 \\
0 & U
\end{array} \right)\,R_{\alpha,\beta,\gamma}^{l}\,\left( \begin{array}{cc}
0 & 0 \\
U^{-1}QU & 0
\end{array} \right)\,R_{\alpha,\beta,\gamma}^{-l}\,\left( \begin{array}{cc}
U^{-1} & 0 \\
0 & U^{-1}
\end{array} \right) =  \nonumber \\
A_{(0,1)}(E)^l \left( \begin{array}{cc} 0 & 0 \\
Q & 0 \end{array} \right) A_{(0,1)}(E)^{-l} \in \gt, \label{alphabeta}
\end{eqnarray}
for any diagonal matrix $Q$. As $U$ is orthogonal,
\begin{equation}\label{tilde_gt}
\gt=\spt \Leftrightarrow \gtt := \left(
\begin{array}{cc}
U & 0 \\
0 & U
\end{array} \right)\gt \left( \begin{array}{cc}
U^{-1} & 0 \\
0 & U^{-1}
\end{array} \right)=\spt
\end{equation}
So we are bring to prove that $\gtt=\spt$. For this, we will use that for every $l\in \Z$ and every diagonal matrix $Q$, by (\ref{alphabeta}),
\begin{equation}\label{tildealphabeta}
R_{\alpha,\beta,\gamma}^{l}\,\left( \begin{array}{cc}
0 & 0 \\
U^{-1}QU & 0
\end{array} \right)\,R_{\alpha,\beta}^{-l} \in \gtt.
\end{equation}

\noindent \textbf{Step 4 :} One can choose $Q=\mathrm{diag}(\sqrt{2},0,\sqrt{2})$ (even for $c_2=0$) to get 
$$U^{-1}QU=\left( \begin{array}{ccc}
0 & -1 & -1 \\
-1 & 0 & 0 \\
-1 & 0 & 0
\end{array}\right).$$
For every $x\in \R$, we set $s(x)=\sin(x)$ and $c(x)=\cos(x)$. By step $3$, for every $l\in \Z$ we have
{\footnotesize $$D_{1}(l):=\left( \begin{array}{cccccc}
0 & \frac{s(l\alpha)c(l\beta)}{\alpha} & \frac{s(l\alpha)c(l\gamma)}{\alpha} & 0 & \frac{s(l\alpha)s(l\beta)}{\alpha \beta} & \frac{s(l\alpha)s(l\gamma)}{\alpha \gamma} \\
\frac{s(l\beta)c(l\alpha)}{\beta} & 0 & 0 & \frac{s(l\alpha)s(l\beta)}{\alpha \beta} & 0 & 0 \\
\frac{s(l\gamma)c(l\alpha)}{\gamma} & 0 & 0 & \frac{s(l\alpha)s(l\gamma)}{\alpha \gamma} & 0 & 0 \\
0 & -c(l\alpha)c(l\beta) & -c(l\alpha)c(l\gamma) & 0 & -\frac{s(l\beta)c(l\alpha)}{\beta} & -\frac{s(l\gamma)c(l\alpha)}{\gamma} \\
-c(l\alpha)c(l\beta) & 0 & 0 & -\frac{s(l\alpha)c(l\beta)}{\alpha} & 0 & 0\\
-c(l\alpha)c(l\gamma) & 0 & 0 & -\frac{s(l\alpha)c(l\gamma)}{\alpha} & 0 & 0
\end{array}\right) $$}
is in $\gtt$. We can also choose $Q=\mathrm{diag}(2,0,2)$ (even if $c_2=0$) to get 
$$U^{-1}QU=\left(\begin{array}{ccc}
2 & 0 & 0 \\
0 & 1 & 1 \\
0 & 1 & 1
\end{array} \right).$$
And so for every $l\in \Z$, 
{\footnotesize $$D_{2}(l):=\left( \begin{array}{cccccc}
\frac{2s(l\alpha)c(l\alpha)}{\alpha} & 0 & 0 & -\frac{2s^{2}(l\alpha)}{\alpha^{2}} & 0 & 0 \\
0 & \frac{s(l\beta)c(l\beta)}{\beta} & \frac{s(l\beta)c(l\gamma)}{\beta} & 0 & -\frac{s^{2}(l\beta)}{\beta^{2}} & -\frac{s(l\beta)s(l\gamma)}{\beta \gamma} \\
0 & \frac{s(l\gamma)c(l\beta)}{\gamma} & \frac{s(l\gamma)c(l\gamma)}{\gamma} & 0 & -\frac{s(l\beta)s(l\gamma)}{\beta \gamma} & -\frac{s^{2}(l\gamma)}{\gamma^{2}} \\
2c^{2}(l\alpha)& 0 & 0 & -\frac{2c(l\alpha)s(l\alpha)}{\alpha} & 0 & 0 \\
0 & c^{2}(l\beta) & c(l\beta)c(l\gamma) & 0 & -\frac{s(l\beta)c(l\beta)}{\beta} & -\frac{s(l\gamma)c(l\beta)}{\gamma} \\
0 & c(l\beta)c(l\gamma) & c^{2}(l\gamma) & 0 & -\frac{s(l\beta)c(l\gamma)}{\beta} & -\frac{s(l\gamma)c(l\gamma)}{\gamma}
\end{array} \right) $$}
is in $\gtt$.

\noindent \textbf{Step 5 :} We prove that except for a discrete set of values of $E$, the matrices $D_1(0),\ldots,D_1(7)$ are linearly independent. Indeed, if one computes the $8\times 8$ determinant of the vectors
$$\left(\begin{array}{c}
-\cos(l\alpha)\cos(l\beta) \\
-\cos(l\alpha)\cos(l\gamma)\\
\frac{\sin(l\alpha)\sin(l\beta)}{\alpha \beta}\\
\frac{\sin(l\alpha)\sin(l\gamma)}{\alpha \gamma}\\
\frac{\sin(l\alpha)\cos(l\beta)}{\alpha}\\
\frac{\sin(l\alpha)\cos(l\gamma)}{\alpha}\\
\frac{\sin(l\beta)\cos(l\alpha)}{\beta}\\
\frac{\sin(l\gamma)\cos(l\alpha)}{\gamma}
\end{array}\right),\  l=0,\ldots,7.$$
representing the $D_1(l)$ matrices, one gets :
\begin{eqnarray}
4096\sin^{4}(\alpha)\sin^{2}(\beta)\sin^{2}(\gamma)(\cos\alpha-\cos\beta)^{4}(\cos^{2}(\alpha)-\cos^{2}(\beta))\times \nonumber \\ 
\times(\cos^{2}(\alpha)-\cos^{2}(\gamma))\left(-\sin^{2}(2\alpha)+\cos^{2}(\beta)+\cos^{2}(\gamma) \right. \nonumber \\
\left. +2\cos\beta\cos\gamma(1-2\cos^{2}(\alpha)) \right)^{2} \label{det88_formula}
\end{eqnarray}
which is a real analytic function of $E$ on $(\sqrt{2},+\infty)$ which does not identically vanish. Thus, this determinant vanishes only on a discrete set $\mathcal{S}_1 \subset (\sqrt{2},+\infty)$.

\noindent \textbf{Step 6 :} Let $E\in (\sqrt{2},+\infty) \setminus \mathcal{S}_1$. By step $5$ all matrices of the form
\begin{equation}\label{forme_mat_88}
\left(\begin{array}{cccccc}
0 & a & b & 0 & g & h \\
c & 0 & 0 & g & 0 & 0 \\
d & 0 & 0 & h & 0 & 0 \\
0 & e & f & 0 & -c & -d \\
e & 0 & 0 & -a & 0 & 0 \\
f & 0 & 0 & -b & 0 & 0 
\end{array}\right)
\end{equation}
for $(a,b,c,d,e,f,g,h)\in \R^{8}$ are in $\gtt$. In particular we set $\tilde{B}_0$ (respectively $\tilde{B}_1$) the matrix of the form (\ref{forme_mat_88}) with $a=1$ and the other parameters equal to $0$ (respectively $b=1$ and the other parameters equal to $0$). Then $\tilde{B}_0\in \gtt$, $\tilde{B}_1 \in \gtt$ and 
\begin{equation}\label{B0_in_gtt}
B_{0}:=[\tilde{B}_{0},D_{1}(0)]=
\left(\begin{array}{cccccc}
0 & 0 & 0 & 0 & 0 & 0 \\
0 & 0 & 0 & 0 & 0 & 0 \\
0 & 0 & 0 & 0 & 0 & 0 \\
0 & 0 & 0 & 0 & 0 & 0 \\
0 & 2 & 1 & 0 & 0 & 0 \\
0 & 1 & 0 & 0 & 0 & 0 
\end{array}\right)\in \gtt.
\end{equation}
and 
\begin{equation}\label{B1_in_gtt}
B_{1}:=[\tilde{B}_{1},D_{1}(0)]=
\left(\begin{array}{cccccc}
0 & 0 & 0 & 0 & 0 & 0 \\
0 & 0 & 0 & 0 & 0 & 0 \\
0 & 0 & 0 & 0 & 0 & 0 \\
0 & 0 & 0 & 0 & 0 & 0 \\
0 & 0 & 1 & 0 & 0 & 0 \\
0 & 1 & 2 & 0 & 0 & 0 
\end{array}\right)\in \gtt.
\end{equation}

\noindent \textbf{Step 7 :} We prove that $(B_0,B_1,D_2(0),\ldots,D_2(10))$ is a family of $13$ linearly independent matrices. For this we have to prove that a $13\times 13$ determinant does not vanish identically. This determinant is given by the $11$ vectors of the $13$ different (and non-colinear) non-zeros elements in $D_2(0),\ldots,D_2(10)$ and the $2$ vectors of the $13$ corresponding coefficients in $B_0$ and $B_1$. Numerically one can verify that $E=1.6 \in (\sqrt{2},+\infty)$, $E=1.6\notin \mathcal{S}_1$ and that for $E=1.6$, the determinant has approximate value of $-3507\neq 0$. Thus it does not identically vanish on $(\sqrt{2},+\infty)$. But this determinant being a real-analytic function of $E$ on $(\sqrt{2},+\infty)$ it is therefore vanishing only on a discrete set of values of $E$, $\mathcal{S}_2 \subset (\sqrt{2},+\infty)$.

\noindent \textbf{Step 8 :} Let $\mathcal{S}_3=\mathcal{S}_1 \cup \mathcal{S}_2$. Let $E\in  (\sqrt{2},+\infty) \setminus \mathcal{S}_3$. Looking at the zero coefficients in $D_1(l)$ and $D_2(l)$ one sees that the families $(D_1(0),\ldots,D_1(7))$ and $(B_0,B_1,D_2(0),\ldots,D_2(10))$ lie in two orthogonal subspaces of $\spt$. Thus, they generate two orthogonal subspaces of dimension $8$ and $13$ and the direct sum of these spaces is still contained in $\gtt$. Thus, $\mathrm{dim}\; \gtt \geq 21$. But $\gtt \subset \spt$ and $\mathrm{dim}\; \spt =21$, so $\gtt=\spt$. It endeed the proof for $E>\sqrt{2}$ as by connectedness, $\clz=\SpT$ for every $E\in (2,+\infty)\setminus \mathcal{S}_3$.
\end{proof}

\begin{proof}[\textbf{Proof of proposition \ref{zariski_f_group} for $E\leq \sqrt{2}$}]
For $E\in(1,\sqrt{2})$ we have the same expression for $A_{(0,1)}(E)$ as (\ref{diagA01}) but with changing $\alpha,\beta, \gamma$ into $\alpha=\sqrt{E-1}$, $\beta=\sqrt{\sqrt{2}-E}$ and $\gamma=\sqrt{E+\sqrt{2}}$. We also change in (\ref{ralphabetagamma}) the $\cos(\beta)$ into a $\cosh(\beta)$, the term $-\beta \sin(\beta)$ into $\beta \sinh(\beta)$ and the term $\frac{1}{\beta} \sin(\beta)$ into $\frac{1}{\beta} \sinh(\beta)$. Then we can follow the proof of the case $E>\sqrt{2}$ which leads to a discrete set $\mathcal{S}_4 \subset (1,\sqrt{2})$ such that for every $E\in (1,\sqrt{2})\setminus \mathcal{S}_4$, $\clz=\SpT$.

\noindent For $E\in(-\sqrt{2},1)$ we set $\alpha=\sqrt{1-E}$, $\beta=\sqrt{\sqrt{2}-E}$ and $\gamma=\sqrt{E+\sqrt{2}}$ and we do the same changes of cosinus and sinus into hyperbolic cosinus and hyperbolic sinus as in the case $E\in(1,\sqrt{2})$, for those involving $\alpha$ and $\beta$. Then following the proof of the case $E>\sqrt{2}$ we get the existence of a discrete set $\mathcal{S}_5 \subset (-\sqrt{2},1)$ such that for every $E\in (-\sqrt{2},1)\setminus \mathcal{S}_5$, $\clz=\SpT$.

\noindent For $E\in (-\infty,-\sqrt{2})$ we set $\alpha=\sqrt{1-E}$, $\beta=\sqrt{\sqrt{2}-E}$ and $\gamma=\sqrt{-E-\sqrt{2}}$ and we change all the sinus and cosinus into hyperbolic sinus and cosinus in (\ref{ralphabetagamma}). It leads to a discrete set $\mathcal{S}_6 \subset (-\infty,-\sqrt{2})$ such that for every $E\in (-\infty,-\sqrt{2})\setminus \mathcal{S}_6$, $\clz=\SpT$.
\vskip 2mm

\noindent Finally if we set $\mathcal{S}_7=\{-\sqrt{2},1,\sqrt{2} \}$, by setting $\mathcal{S}=\mathcal{S}_3 \cup \ldots \cup \mathcal{S}_7$ we have a discrete set such that for every $E\in \R\setminus \mathcal{S}$, $G(E)$ is Zariski dense in $\SpT$, which proves proposition \ref{zariski_f_group}.
\end{proof}

\noindent To finish the proof of theorem \ref{th1} it just remains to check that Kotani's theory can be applied for $H(\omega)$ which will be done in Section \ref{sec:proofth4}. Then from the non-vanishing of the Lyapunov exponents outside of $\mathcal{S}$, which is of Lebesgue measure $0$, we deduce the absence of absolutely continuous spectrum of $H(\omega)$. And thus theorem \ref{th1} is proved.

\noindent In our study of the Lyapunov exponents it would be also interesting to look at what happens at the values of $E$ in $\mathcal{S}$. We do not know exactly if the Lyapunov exponents are all vanishing at these energies or if only some of them vanish or even if they are all positive but not distinct. To handle directly with the computation of the Lyapunov exponents at the values of $E$ in $\mathcal{S}$ is much more difficult than in the case of scalar-valued operators. For the scalar-valued Anderson model in \cite{DLS06}, the Lyapunov exponent vanishes at the critical energies, which is clear because at these values, $G(E)$ is compact. Here it is no longer the case and except doing some numerical attempts to understand the situation, no rigorous proof has been found yet. We think to be in the presence of energies at which some Lyapunov exponents vanish while others are not. The case were all the Lyapunov exponents vanish does not seem to happen.

%-------------------------------------------------------------%
%:2.Proof of theorems 2,3,4 and 5.
%-------------------------------------------------------------%
\section{Regularity results on Lyapunov exponents and the IDS}\label{sec:regular}

\subsection{H\"older continuity of the Lyapunov exponents}\label{sec:regularlyap}

In this section we prove theorem \ref{th2}. We will see how to use general results of \cite{boumazarmp} for the operator $H(\omega)$. 

\noindent We can deduce theorem \ref{th2} from theorem $1$, p.$885$ in \cite{boumazarmp} once we have proved the following estimates on the transfer matrices.

\begin{prop}\label{mat_transfer_estim}
Let $I$ be a compact interval in $\R$. There exist $C_{1}>0$, $C_{2}>0$ independent of $n,\omega^{(n)},E$ such that for every $p\in \{1,\ldots ,N\}$,
\begin{equation}\label{estim1}
||\wedge^{p} A_{(n,n+1]}^{\omega^{(n)}}(E)||^{2} \leq \exp( pC_{1} + p|E| +p) \leq C_{2}.
\end{equation}
And there exists $C_{3}>0$ independent of $n,\omega^{(n)},E$ such that for every $E,E'\in I$ and every $p\in \{1,\ldots ,N\}$, 
\begin{equation}\label{estim2}
||\wedge^{p} A_{(n,n+1]}^{\omega^{(n)}}(E)-\wedge^{p} A_{(n,n+1]}^{\omega^{(n)}}(E')||\leq C_{3} |E-E'|.
\end{equation}
\end{prop}
\vskip 2mm

\begin{proof}
First we recall that $A_{(0,1)}(E)$ is obtained by solving the free system $H(\omega)u=Eu$ on $(0,1)$. This system is of the same form as the one from which we deduce the transfer matrices in \cite{boumazarmp}. Thus, the estimates (\ref{estim1}) and (\ref{estim2}) which were proved for the transfer matrices in \cite{boumazarmp} are still valid for $A_{(0,1)}(E)$ and thus $A_{(0,1)}(E)$ verifiy (\ref{estim1}) and (\ref{estim2}). Then $\supp \nu$ being bounded, $||M(\mathrm{diag}(c_1\omega_{1}^{(n)}, \ldots,c_N\omega_{N}^{(n)}))||$ will be bounded uniformly in $n$, $\omega^{(n)}$ and also $E$ because it does not depend on $E$. Let $C>0$ be independent of $n$, $\omega^{(n)}$ and $E$ such that $||M(\mathrm{diag}(c_1\omega_{1}^{(n)},\ldots,c_N\omega_{N}^{(n)})) || \leq C$. Then $||\wedge^p M(\mathrm{diag}(c_1\omega_{1}^{(n)},\ldots,c_N\omega_{N}^{(n)})) || \leq C^p$. Using the fact that $||\wedge^p (MN) || \leq ||\wedge^p M ||\; ||\wedge^p N ||$ for every matrices and the relation (\ref{transfer_mat_split}), we finally obtain  (\ref{estim1}) and (\ref{estim2}) for $A_{(n,n+1]}^{\omega^{(n)}}(E)$.
\end{proof}

\noindent Applying theorem $1$ in \cite{boumazarmp} and proposition \ref{mat_transfer_estim} we have proved theorem \ref{th2}. Then applying theorem \ref{th2} on every compact interval $I\subset \R \setminus \mathcal{S}$ where $\mathcal{S}$ is obtained in theorem \ref{th1}, we get theorem \ref{th3}.

\subsection{H\"older continuity of the IDS}\label{sec:regularids}

\subsubsection{Existence of the IDS}\label{sec:existids}

\noindent Once again we will follow the method used in \cite{boumazarmp} to prove this time the existence of the integrated density of states associated to $H(\omega)$ and its H\"older continuity. The proof of the existence of the limit (\ref{IDSdefi}) is based upon the fact that the one-parameter semi-group $(e^{-tH^{(D)}(\omega)})_{t>0}$ admits an integral kernel in $L^{2}(D^2)\otimes \mathcal{M}_{\mathrm{N}}(\C)$. This kernel is coming from a Feynman-Kac formula. There already exists such formula for scalar-valued point interactions operators as presented in \cite{FM2003}. Adapting the Borel measure representation method of \cite{FM2003}, using Lie-Trotter formula as it is done in \cite{boumazarmp} and noticing that the time-ordered exponential in \cite{boumazarmp} becomes now a usual exponential, one gets
\begin{equation}\label{feynman11}
\forall f\in L^{2}(\R)\otimes \C^N,\ \forall x\in \R,\ e^{-tH(\omega)}f(x)=\int_{\R} K_{t}(x,y)f(y)\mathrm{d}y
\end{equation}
with
\begin{equation}\label{feynman12}
\forall t>0,\ \forall x,y\in \R,\ K_{t}(x,y)=\int e^{-tV_0} \mathrm{d}\mu_{x,y,\omega}(\mathsf{w})
\end{equation}
where for every $\omega \in \tilde{\Omega}$ fixed, $\mu_{x,y,\omega}$ is a finite measure on the space $\mathsf{W}_{x,y}$ of the continuous paths $\mathsf{w}$ on $[0,t]$ such that $\mathsf{w}(0)=x$ and $\mathsf{w}(1)=y$. Then to deduce the kernel of $e^{-tH^{(D)}(\omega)}$, we introduce $T_{D}(\mathsf{w})$, the time of first exit from $D$ of the path $\mathsf{w}\in \mathsf{W}_{x,y}$,
\begin{equation}\label{timeexit}
T_D(\mathsf{w})=\inf \{ t>0\ |\ \mathsf{w}(t) \notin D\}.
\end{equation}
Then we have (see \cite{Kni81})
\begin{equation}\label{feynman21}
\forall f\in L^{2}(D)\otimes \C^N,\ \forall x\in D,\ e^{-tH^{(D)}(\omega)}f(x)=\int_{D} K_{t}^{(D)}(x,y)f(y)\mathrm{d}y
\end{equation}
with
\begin{equation}\label{feynman22}
\forall t>0,\ \forall x,y\in D,\ K_{t}^{(D)}(x,y)=\frac{e^{-\frac{|x-y|^2}{2t}}}{\sqrt{2\pi t}} \int \chi_{\{ t<T_D(\mathsf{w})\} }e^{-tV_0} \mathrm{d}\mu_{x,y,\omega}(\mathsf{w}).
\end{equation}
As we can see, for every $t>0$, $K_{t}^{(D)}$ is in $L^{2}(D^2)\otimes \mathcal{M}_{\mathrm{N}}(\C)$. Thus all the discussion made in Section $2.3$ of \cite{boumazarmp} applies here to get the existence of the IDS for every $E\in \R$ and its realization as the distribution function of a measure $\mathfrak{n}$ called the \emph{density of states}.

\subsubsection{Kotani's theory and proof of theorem \ref{th4}}\label{sec:proofth4}

\noindent We start by adapting Kotani's theory of \cite{KS88} to our setting. According to the presentation made in \cite{boumazarmp} we actually only have to prove that theorem $2.1$ (a) of \cite{KS88} is true for the operator $H(\omega)$. We fix $\omega \in \tilde{\Omega}$. Let $\C_{+}=\{ z\in \C,\ \mathrm{Im} z >0 \}$ and $\C_{-}=\{ z\in \C,\ \mathrm{Im} z < 0 \}$. For $E\in \C_{+} \cup \C_{-}$ we set 
$$J_{+}(E)=\left\lbrace f\in \bigoplus_{i=1}^N D(H_{\omega_i})\ \big| \ H(\omega)f=Ef\ \mathrm{and}\ \int_{0}^{\infty} |f(x)|^{2}dx <+\infty \right\rbrace$$
and
$$J_{-}(E)=\left\lbrace f\in \bigoplus_{i=1}^N D(H_{\omega_i})\ \big| \ H(\omega)f=Ef\ \mathrm{and}\ \int_{-\infty}^{0} |f(x)|^{2}dx<+\infty  \right\rbrace.$$

\begin{prop}\label{dimJ}
We have : $\mathrm{dim}\ J_{+}(E)=\mathrm{dim}\ J_{-}(E)=N$.
\end{prop}
\vskip 2mm

\begin{proof}
Each $H_{\omega_i}$ is in the limit point case of singular Weyl theory (see \cite{albeverio}) which ensure the existence of a unique solution $f_{1,i}\in D(H_{\omega_i})$ of $H_{\omega_i}f=Ef$ such that $\int_{0}^{\infty} |f_{1,i}(x)|^{2}dx <+\infty$ and a unique solution $f_{2,i}\in D(H_{\omega_i})$ of $H_{\omega_i}f=Ef$ such that $\int_{-\infty}^{0} |f_{2,i}(x)|^{2}dx <+\infty$. Then if we introduce $H_{\omega_i}^+=H_{\omega_i}|_{D(H_{\omega_i})\cap L^2(0,+\infty)}$, $\mathrm{dim}\ \mathrm{Ker}(H_{\omega_i}^+ -E)=1$. With $H_{\omega_i}^-=H_{\omega_i}|_{D(H_{\omega_i})\cap L^2(-\infty,0)}$, we also have $\mathrm{dim}\ \mathrm{Ker}(H_{\omega_i}^- -E)=1$. Then
$$\mathrm{dim}\ \mathrm{Ker}\left(\bigoplus_{i=1}^{N} H_{\omega_i}^+ -E \right) =\mathrm{dim}\ \mathrm{Ker}\left(\bigoplus_{i=1}^{N} H_{\omega_i}^- -E \right)  =N.$$
As $V_0$ is bounded, from deficiency index theory we get 
$$\mathrm{dim}\ J_{+}(E) =\mathrm{dim}\ \mathrm{Ker}\left(\bigoplus_{i=1}^{N} H_{\omega_i}^+ +V_0 -E \right) = \mathrm{dim}\ \mathrm{Ker}\left(\bigoplus_{i=1}^{N} H_{\omega_i}^+ -E \right)=N$$
and the same goes for $J_{-}(E)$.
\end{proof}

\noindent From this proposition we get as in \cite{KS88}, Corollary $2.2$, the following result.

\begin{prop}\label{Fplus}
Let $E\in \C_{+}\cup \C_{-}$ and $\omega \in \tilde{\Omega}$. Then there exists a unique function  $x\mapsto F_{+}(x,E)$ with values in $\mathcal{M}_{\mathrm{N}}(\C)$ (respectively $x\mapsto F_{-}(x,E)$) satisfying  
$$H(\omega)F_{+}=EF_{+},\ F_{+}(0,E)=I,\ \mathrm{and}\ \int_{0}^{\infty} ||F_{+}(x,E)||^{2}dx <+\infty,$$
respectively  
$$H(\omega)F_{-}=EF_{-},\ F_{-}(0,E)=I,\ \mathrm{and}\ \int_{-\infty}^{0} ||F_{-}(x,E)||^{2}dx <+\infty.$$
\end{prop}

\noindent This proposition is the starting point of all the theory on the Floquet exponent $w$ and the $M_{\pm}$-functions presented in \cite{KS88}. Thus Kotani's theory on the absolutely continuous spectrum apply for $H(\omega)$ and it finishes the proof of theorem \ref{th1}. We also recall that \cite{KS88} only considers $\R$-ergodic operators whereas $H(\omega)$ is $\Z$-ergodic. To avoid this difficulty we can refer to the suspension procedure developed by Kirsch in \cite{kirsch}. This procedure allows us to construct from $H(\omega)$ an operator $\hat{H}(\hat{\omega})$ define on a bigger probability space which is $\R$-ergodic. $\hat{H}(\hat{\omega})$ is also constructed in a way such that its IDS and Lyapunov exponents exist if and only if those of $H(\omega)$ exist and in this case they are equal for both operators.

\noindent Also, from the properties of the Floquet exponent $w$ combined with previous results of Kotani (see \cite{kotani83}) one can repeat the discussions of Sections $4.1$ and $4.2$ in \cite{boumazarmp} to prove the following Thouless formula for $H(\omega)$.

\begin{prop}\label{thouless}
For Lebesgue-almost every $E\in \R$ we have  
\begin{equation}\label{thoulesseq}
(\gamma_{1}+\ldots+\gamma_{N})(E)= -\alpha+\int_{\R} \log \left( \left| \frac{E'-E}{E'-\mathrm{i}}\right| \right)~\dd \mathfrak{n}(E')
\end{equation}
where $\alpha$ is a real number independent of $E$ and $\mathfrak{n}$ is the density of states of $H(\omega)$. Moreover, if  $I\subset \R$ is an interval on which $E\mapsto (\gamma_{1}+\ldots +\gamma_{N})(E)$ is continuous then (\ref{thoulesseq}) holds for every $E\in I$.
\end{prop}

\noindent Using this Thouless formula, theorem \ref{th2} and properties of the Hilbert transform (see \cite{neri}), we can obtain the proof of theorem \ref{th4} exactly in the same way as it was done at theorem $4$ of \cite{boumazarmp}, Section $4.3$. Then applying theorem \ref{th4} on any compact interval $I\subset \R \setminus \mathcal{S}$ where $\mathcal{S}$ is given in theorem \ref{th1}, we prove theorem \ref{th5}.

\vskip 2mm

\noindent As we can see, the methods to prove regularity of Lyapunov exponents and regularity of the IDS for the point interactions model $H(\omega)$ are completely similar to those for matrix-valued Anderson models. The main differences are to be found in the proof of the Zariski denseness of the F\"urstenberg group. Indeed, in \cite{boumazampag} we proved Zariski denseness of the F\"urstenberg group of an Anderson operator acting on $L^2(\R)\otimes \C^2$ using algebraic technics different than those used in the proof of theorem \ref{th1} in the present paper.

%-------------------------------------------------------------%

%-------------------------------------------------------------%

%-------------------------------------------------------------%

%-------------------------------------------------------------%
\end{document}